\documentclass[british,english]{article}
\usepackage[T1]{fontenc}
\usepackage[latin9]{inputenc}
\usepackage{geometry}
\usepackage{amsthm}
\usepackage{amsmath}
\usepackage{amssymb}
\usepackage{color}
\usepackage{hyperref}
\usepackage{svg}
\usepackage{calc}
\usepackage{graphicx}
\usepackage{epstopdf}
\usepackage{pdfpages}
\usepackage{mathtools}
\usepackage{authblk}

\newcommand{\lyxaddress}[1]{
\par {\raggedright #1
\vspace{1.4em}
\noindent\par}
}

\def\mc {\mathcal}

\def\mr {\mathrm}
\def\mb {\mathbb}

\makeatletter
\theoremstyle{plain}
\newtheorem{thm}{\protect\theoremname}
  \theoremstyle{plain}
  \newtheorem{lem}[thm]{\protect\lemmaname}
  \theoremstyle{plain}
  \newtheorem{fact}[thm]{\protect\factname}
  \theoremstyle{remark}
  \newtheorem{rem}[thm]{\protect\remarkname}
    \theoremstyle{plain}
  \newtheorem{coll}[thm]{\protect\collname}
      \theoremstyle{plain}
  \newtheorem{definition}[thm]{\protect\defname}

\makeatother

\usepackage{babel}
  \addto\captionsbritish{\renewcommand{\factname}{Fact}}
    \addto\captionsbritish{\renewcommand{\collname}{Corollary}}
  \addto\captionsbritish{\renewcommand{\lemmaname}{Lemma}}
  \addto\captionsbritish{\renewcommand{\remarkname}{Remark}}
  \addto\captionsbritish{\renewcommand{\theoremname}{Theorem}}
  \addto\captionsbritish{\renewcommand{\defname}{Definition}}
  \addto\captionsbritish{\renewcommand{\theoremname}{Theorem}}
  \addto\captionsenglish{\renewcommand{\factname}{Fact}}
  \addto\captionsenglish{\renewcommand{\lemmaname}{Lemma}}
  \addto\captionsenglish{\renewcommand{\remarkname}{Remark}}
  \addto\captionsenglish{\renewcommand{\theoremname}{Theorem}}
  \addto\captionsenglish{\renewcommand{\defname}{Definition}}
  \addto\captionsenglish{\renewcommand{\theoremname}{Theorem}}
  
  \providecommand{\factname}{Fact}
  \providecommand{\lemmaname}{Lemma}
  \providecommand{\remarkname}{Remark}
\providecommand{\theoremname}{Theorem}
\providecommand{\collname}{Corollary}
\providecommand{\defname}{Definition}
\begin{document}
\selectlanguage{british}%
\global\long\global\long\global\long\def\bra#1{\mbox{\ensuremath{\langle#1|}}}
 \global\long\global\long\global\long\def\ket#1{\mbox{\ensuremath{|#1\rangle}}}
 \global\long\global\long\global\long\def\bk#1#2{\mbox{\ensuremath{\ensuremath{\langle#1|#2\rangle}}}}
 \global\long\global\long\global\long\def\kb#1#2{\mbox{\ensuremath{\ensuremath{\ensuremath{|#1\rangle\!\langle#2|}}}}}
\title{Products of finite  order rotations and quantum gates universality}
\author{Katarzyna Karnas$^{1}$ and Adam Sawicki$^{1}$}
\selectlanguage{english}
\date{}
\maketitle
\lyxaddress{$^1$Center for Theoretical Physics PAS, Al. Lotnik\'ow 32/46, 02-668, Warszawa, Poland}
\begin{abstract}
We consider a product of two finite order quantum $SU(2)$-gates $U_1$, $U_2$ and ask when $U_1\cdot U_2$ has an infinite order. Using the fact that $SU(2)$ is a double cover of $SO(3)$ we actually study the product  $O(\gamma,\vec{k}_{12})$ of two rotations $O(\phi,\vec{k}_1)\in SO(3)$ and $O(\phi,\vec{k}_2)\in SO(3)$ about axes $\vec{k}_1$, $\vec{k}_2\in \mathbb{R}^3$. In particular we focus on the case when $\vec{k}_1\cdot\vec{k}_2=0$, and  $\phi_1=\phi=\phi_2$ are rational multiple of $\pi$ and show that  $\gamma$ is not a rational multiple of $\pi$ unless $\phi\in\{\frac{k\pi}{2}:k\in\mathbb{Z}\}$. The proof presented in this paper boils down to finding all pairs $\gamma,\phi\in \{a\pi : a\in\mathbb{Q}\}$  that are solutions of  $\cos\frac{\gamma}{2}=\cos^2\frac{\phi}{2}$. \end{abstract}

\section{Introduction}\label{sec:intro} 


A finite subset $\mc{S}=\{g_1,\ldots,g_k\}$ of a Lie group $G$ is called a {\it generating set} if the set 
\begin{gather}
<\mathcal{S}>:=\{g_{i_1}^{k_1}\cdot\ldots\cdot g_{i_m}^{k_m}:g_{i_j}\in \mathcal{S},k_j\in \mathbb{N},i_j\in\{1,\ldots,n\}\},
\end{gather}
consisting of all words built using elements of $\mathcal{S}$ is dense in $G$. The problem of characterisation of generating sets for various types of Lie groups has recently attracted much more interest, in particular due to its direct connection to the quantum gates universality \cite{barenco,aronson,oszmanoec,exp2,Reck,sawicki2,sawicki3} (c.f. \cite{childs,daniel} for a similar problem of generating Hamiltonians). In short words the set of quantum gates is universal if and only if it is a generating set for $SU(d)$ or $SO(d)$. It is worth mentioning that for abelian $G$ the smallest generating set consist of one element whereas for nonabelian it is typically enough to use two elements \cite{kuranishi1}. 


Although $<\mc{S}>$ is closed under group multiplication it might not be a group itself as the inverses of some elements may not be in $<S>$. Nevertheless, its closure has a group structure (see fact 2.6 in \cite{sawicki2}). As was shown by Schur in his solution to the Burnside problem (see lemma 36.2 of \cite{curtis}), infinite but finitely generated groups of matrices over $\mb{R}$ or $\mb{C}$ always contain infinite order elements. Note that the necessary condition for $\mathcal{S}$ to be universal is that $<\mathcal{S}>$ contains infinite number of elements. Verification of this can be thus accomplished by pointing one infinite order element in $<\mathcal{S}>$. The efficiency of this approach, however, is limited when all generators (elements of $\mathcal{S}$) are of the finite order. In this case one should consider products of elements from  $\mathcal{S}$, as the noncommutativity of $G$ implies that the composition of two (or more) finite order elements may be of infinite order. In fact, the main step in the proof contained in \cite{NC00} of the universality for gates $\{H,T\left(\frac{\pi}{8} \right) \}\subset SU(2)$, where $H$ is the Hadamard gate and $T_{\frac{\pi}{8}}$ is a phase gate (both are elements of finite order), i.e. 
\begin{gather}
H=\frac{i}{\sqrt{2}}\left(\begin{array}{cc}1&1\\1&-1 \end{array} \right),\,\,T_{\frac{\pi}{8}}=\left(\begin{array}{cc}e^{i\pi/8}&0\\0&e^{-i\pi/8}\end{array} \right),
\end{gather}
is showing that $U=H\cdot T_{\frac{\pi}{8}}$ is of infinite order. The method used in the proof boils down to finding the minimal polynomial for $e^{i\phi}$, where the spectrum of $U$ is given by $\{e^{i\phi}, e^{-i\phi}\}$. If this polynomial is not a cyclotomic polynomial then $
\phi$ is not a rational multiple of $\pi$ and $U$ is of infinite order. Although the application of this method for gates $H$ and $T_{\frac{\pi}{8}}$ is a simple exercise \cite{Boykin} it becomes much more difficult for $\{H,T_\phi\}$, where $\phi$ is any rational multiple of $\pi$ or for an arbitrary pair of $SU(2)$ gates. In this short paper we discuss cases when one can efficiently prove that the product of two finite order elements from $SU(2)$ is of infinite order. As the group $SU(2)$ is the double cover of $SO(3)$ we work with matrices belonging to $SO(3)$ rather than to $SU(2)$. In particular we focus on the case when $\vec{k}_1\cdot\vec{k}_2=0$, and  $\phi_1=\phi=\phi_2$ are rational multiples of $\pi$ and show that  $\gamma$ is not a rational multiple of $\pi$ unless $\phi\in\{\frac{k\pi}{2}:k\in\mathbb{Z}\}$. The result obtained for this case is presented as our main theorem whereas, to make our notation clearer, we call all the previously known results as facts and our auxiliary results by lemmas. Our proof boils down to  finding all pairs $\gamma,\phi\in \{a\pi : a\in\mathbb{Q}\}$  that are solutions of  $\cos\frac{\gamma}{2}=\cos^2\frac{\phi}{2}$. In order to show that $\gamma$ is not a rational multiple of $\pi$ one can for example find the minimal polynomial for $e^{i\gamma}$ and check if this polynomial is cyclotomic. This approach was studied in \cite{sawicki1}, where in section 3.1 the author shows that in general the problem is intractable. Another approach that one could follow is finding the minimal polynomial for $\cos^2\phi$ and showing that its resultant with any of the Chebyshev polynomials in nonzero. This, however, again turns out to be hopeless calculation. Our approach consists of the following steps:
\begin{enumerate}
\item Write equation $\cos\frac{\gamma}{2}=\cos^2\frac{\phi}{2}$ as $2\cos\frac{\gamma}{2}=1+\cos\phi$.
\item For $\phi$ - a rational multiple of $\pi$ show that at least one coefficient of the minimal polynomial of $\cos\phi$ is noninteger (see lemma \ref{fact:coeffs_minpoly_cos}).
\item Prove that the minimal polynomial for $2\cos\frac{\gamma}{2}$ has integer coefficients when $\gamma$ is rational multiple of $\pi$ (see lemma \ref{fact:minpoly_2cos_coeffs}).
\item Using the companion matrix formalism described in section \ref{sec:comp_matrix} find formulas for coefficients of the minimal polynomial of $1+\cos\phi$ in terms of the coefficients of the minimal polynomial for $\cos\phi$ (see section \ref{sec:proofs}).
\item Show that coefficients of the minimal polynomial for $1+\cos\phi$ are not all integers if $\phi\notin\{\frac{k\pi}{2}:k\in\mathbb{Z}\}$ (see section \ref{sec:proofs}).
\end{enumerate}
The presented approach is both direct and simple. As we show all the above steps involve easy calculations. In the last section we discuss some cases when the rotation angles are not equal and the rotation axes are not perpendicular.


\section{Results} \label{sec:results}

Let $O(\phi,\vec{k})$ be a rotation by $\phi$ around the axis $\vec{k}$. We will consider three rotations about axes $\vec{k}_x,\vec{k}_y$ and $\vec{k}_z$ all by the same angle $\phi$:
\begin{gather}\label{so3_matrices_def}
O(\phi,\vec{k}_x)=\left(\begin{array}{ccc}1&0&0\\0&\cos\phi&\sin\phi\\0&-\sin\phi&\cos\phi\end{array} \right),\:O(\phi,\vec{k}_y)=\left(\begin{array}{ccc}\cos\phi&0&\sin\phi\\0&1&0\\-\sin\phi&0&\cos\phi\end{array} \right),\:O(\phi,\vec{k}_z)=\left(\begin{array}{ccc}\cos\phi&\sin\phi&0\\-\sin\phi&\cos\phi&0\\0&0&1\end{array} \right).
\end{gather}
The corresponding $SU(2)$ matrices are:
\begin{gather}\label{su2_matrices_def}
U(\frac{\phi}{2},\vec{k}_x)=\left(\begin{array}{cc}
\cos\frac{\phi}{2}&\sin\frac{\phi}{2}\\-\sin\frac{\phi}{2}&\cos\frac{\phi}{2}\end{array} \right),\:U(\frac{\phi}{2},\vec{k}_y)=\left(\begin{array}{cc}\cos\frac{\phi}{2}&i\sin\frac{\phi}{2}\\i\sin\frac{\phi}{2}&\cos\frac{\phi}{2}\end{array} \right),\:U(\frac{\phi}{2},\vec{k}_z)=\left(\begin{array}{cc}e^{i\frac{\phi}{2}}&0\\0&e^{-i\frac{\phi}{2}}\end{array} \right).
\end{gather}
The product of any two of the above three rotations (\ref{so3_matrices_def}) is again a rotation by an angle $\gamma$ where:
\begin{gather}\label{eq1}
2\cos\gamma+1=\cos^2\phi+2\cos\phi.
\end{gather}
Using  trigonometric identities we can write (\ref{eq1}) in a simpler form:
\begin{gather}\label{gamma_phi_transform1}
\pm \cos\frac{\gamma}{2}=\cos^2\frac{\phi}{2}.
\end{gather}
Note that the same equation is obtained from the product of any two unitary matrices (\ref{su2_matrices_def}), i.e. the resulting matrix has spectrum $\{e^{-i\frac{\gamma}{2}},e^{i\frac{\gamma}{2}}\}$, where $\gamma$ is determined by equation (\ref{gamma_phi_transform1}). Our main technical result is 
\begin{thm}\label{thm:main}
Assume $\phi$ is a rational multiple of $\pi$. Then $\gamma$  given by (\ref{gamma_phi_transform1}) is a rational multiple of $\pi$ if and only if $\phi\in\{\frac{k\pi}{2}:k\in \mathbb{Z}\}$.
\end{thm}
Using theorem \ref{thm:main}, one easily deduces the following:
\begin{coll}\label{lem:productRotations}
Assume $\phi$ is a rational multiple of $\pi$ and  $\phi\notin\{\frac{k\pi}{2}: k\in\mathbb{Z}\}$. Let $O(\phi,\vec{k}_1),O(\phi,\vec{k}_2)\in SO(3)$ be rotations around orthogonal axes $\vec{k}_i\perp\vec{k}_j$, by the angle $\phi$. Then $O(\gamma,\vec{k})=O(\phi,\vec{k}_1)O(\phi,\vec{k}_2)$ is the rotation by an angle $\gamma$ which is an irrational multiple of $\pi$. Moreover, the set generated by $\{O(\phi,\vec{k}_1),O(\phi,\vec{k}_2)\}$ or by the corresponding unitary matrices $\{U(\phi/2\vec{k}_1),U(\phi/2,\vec{k}_2)\}$ is infinite if and only if $\phi \notin \{\frac{k\pi}{2}:k\in \mathbb{Z}\}$. 
\end{coll}
\begin{proof}
We only need to verify the case $\phi \in \{\frac{k\pi}{2}:k\in \mathbb{Z}\}$. But in this case matrices (\ref{so3_matrices_def}) have entries in  $\{1,-1\}$ and generate either permutation group $S_3$ or finite abelian group. 
\end{proof}
%

\section{Number fields and minimal polynomials}\label{app:minpolys}

In this section we present basic facts concerning minimal polynomials. Using the {\it companion matrix} formalism we explain the relationship between the minimal polynomials of two algebraic numbers and the minimal polynomial of their sum and product.

Recall that $\alpha\in\mathbb{C}$ is an \emph{algebraic number} if it is a root of a polynomial $p\in \mathbb{Q}[x]$, i.e. a polynomial with rational coefficients. The order of $\alpha$ is the order of the monic polynomial $m_\alpha\in\mathbb{Q}[x]$ of the least degree, having $\alpha$ as a root. In other words $m_\alpha$ is irreducible over $\mb{Q}$, i.e. it cannot be decomposed into a product of polynomials from $\mathbb{Q}[x]$. This monic polynomial is called the minimal polynomial of $\alpha$ and is uniquely determined by $\alpha$. If $\alpha$ is not algebraic then it is called \emph{transcendental}. Algebraic numbers form an infinite countable set, $\mb{A}$ \cite{baker}, whereas the set $\mb{T}$ of transcendental numbers is uncountable and we have $\mb{A}\cup \mb{T}=\mb{C}$. The examples of algebraic numbers are all integers, rational numbers or the numbers of the form: $\sqrt[k]{n}$, $\cos\frac{2k\pi}{n}$, $e^{2ik\pi/n}$, $\sin\frac{2k\pi}{n}$, where $n,k\in\mb{N}$. On the other hand the numbers $\pi,e$ are known to be transcendental. Obviously if $a\in\mathbb{A}\backslash \{0\}$ and $t\in\mb{T}$, then also $a^{-1},-a$ are algebraic numbers and $t^{-1},-t$ are transcendental numbers.


Having defined algebraic and transcendental numbers, we introduce the notion of a field extension. Let $\mathbb{L}$ be a field that contains $\mathbb{Q}$ as a subfield and let $S$ be a subset of $\mathbb{L}$. We define the field $\mathbb{Q}(S)$ to be the smallest field that contain both $\mathbb{Q}$ and $S$. We will call $\mathbb{Q}(S)$ a field extension of $\mathbb{Q}$ obtained by adjoining elements of $S$. If all elements in $\mathbb{Q}(S)$ are algebraic numbers then the field $\mathbb{Q}(S)$ is called an algebraic field extension. One can view $\mathbb{Q}(S)$ as a vector space over $\mathbb{Q}$. The dimension of this space, denoted by $[\mathbb{Q}(S):\mathbb{Q}]$, will be called a degree of $\mathbb{Q}(S)$ over $\mathbb{Q}$. If $[\mathbb{Q}(S):\mathbb{Q}]<\infty$ the extension is the finite extension. Finite extensions are algebraic as for any $\alpha\in \mathbb{Q}(S)$ numbers $1,\alpha,\alpha^2,\ldots,\alpha^{[\mathbb{Q}(S):\mathbb{Q}]}$ are linearly dependent over $\mathbb{Q}$ which means there is a polynomial $p\in \mathbb{Q}[x]$ which satisfies $p(\alpha)=0$. The converse is not true.

In this paper we will deal with extensions $Q(S)$, where $S=\{\alpha\}$ and $\alpha$ is an algebraic number of order $n$ whose minimal polynomial is $m_\alpha\in\mathbb{Q}[x]$. For convenience the field $\mathbb{Q}(S)$ will be denoted by $\mathbb{Q}(\alpha)$. Among elements  belonging to $\mathbb{Q}(\alpha)$ are all polynomial expressions in $\alpha$. It turns out that these expressions already form a field which contains $\alpha$ and $\mathbb{Q}$ and therefore they constitute $\mathbb{Q}(\alpha)$. We note that for all $p\in \mathbb{Q}[x]$ satisfying $p(\alpha)=0$ the corresponding expression is $0$. Thus $\mathbb{Q}(\alpha)$ is isomorphic $\mathbb{Q}[x]/(m_\alpha)$, where $(m_\alpha)$ is an ideal of polynomials vanishing on $\alpha$. Using the polynomial division formula, any element $\beta\in\mathbb{Q}(\alpha)$ can be written as $\beta=m_\alpha(\alpha)f(\alpha)+r(\alpha)$, where $\mathrm{deg}\;r<\mathrm{deg}\;m_\alpha=n$. Thus any element of $Q(\alpha)$ is a polynomial in $\alpha$ of degree less than $n$ with coefficients in $\mathbb{Q}$. Therefore $Q(\alpha)$ is a finite extension whose basis is $\{1,\alpha,\ldots,\alpha^{n-1}\}$ and $[\mathbb{Q}(\alpha):\mathbb{Q}]=n$. Hence any element $\beta \in \mathbb{Q}(\alpha)$ is algebraic and $\mathbb{Q}(\beta)\subset\mathbb{Q}(\alpha)$. One can show that the order of $\beta$ is a divisor of the order of $\alpha$, i.e. a divisor of $n$ and is given by 
\begin{gather}\label{beta-deg}
\mr{deg}\;m_\beta=[\mb{Q}(\beta):\mathbb{Q}]=[\mb{Q}(\alpha):\mb{Q}]/[\mb{Q}(\alpha):\mb{Q}(\beta)],
\end{gather}where $m_\beta$ is the minimal polynomial for $\beta$.

Consider field extension $Q(e^{i\phi})$ where $\phi$ is a rational multiple of $\pi$. As  $e^{i\phi}$ is a root of unity, $Q(e^{i\phi})$ is a finite algebraic extension. Next, $\sin\phi$, $\cos\phi$ depend on $e^{i\phi}$ as $\cos\phi=\frac{e^{i\phi}+e^{-i\phi}}{2}$ and $\sin\phi=\frac{e^{i\phi}-e^{-i\phi}}{2i}$ and therefore $\mathbb{Q}(\cos\phi)\subset \mathbb{Q}(e^{i\phi})$ and $\mathbb{Q}(\sin\phi)\subset \mathbb{Q}(e^{i\phi})$. The minimal polynomials of these functions are characterised by several unique properties. In this paper we place particular emphasis on the minimal polynomials $\psi_n$ for $\cos\frac{2\pi}{n}$ and $\eta_n$ for $2\cos\frac{\phi}{2}$. Their properties will be crucial in the proof of theorem \ref{thm:main}, therefore we describe them at length in section \ref{app:cosines}.

\subsection{Companion matrix formalism}\label{sec:comp_matrix}

The notion of a minimal polynomial can be generalized to  matrices over $\mb{Q}$. We say that a square matrix $M\in \mr{M}_{n}(\mb{Q})$ is a root of a polynomial $p\in\mathbb{Q}[x]$ if $p(M)=0$ (in other words $p$ \emph{annihilates} $M$). Let $\chi_M\in\mathbb{Q}[x]$ be the characteristic polynomial of $M\in \mr{M}_{n}(\mb{Q})$. By the Cayley-Hamilton theorem $M$ is a root of its own characteristic polynomial, i.e. $\chi_M(M)=0$. The monic polynomial $m_M\in\mathbb{Q}[x]$ of the smallest degree that is irreducible over $\mathbb{Q}$ and annihilates $M$ will be called the \emph{minimal polynomial} of $M$. Conversely, to every  minimal polynomial $m_\alpha\in\mathbb{Q}[x]$  we can associate the matrix called a \emph{companion matrix} $M_\alpha$, defined as follows:
\begin{definition}\label{def:comp_matrix}
Let $m_\alpha$  be a minimal polynomial of $\alpha$ of degree $\deg\; m_\alpha=d$ given by $m_\alpha(x) = x^d+\sum_{k=0}^{d-1} c_k\cdot x^k$. The \emph{companion matrix} $M_\alpha$ is the $n\times n$ matrix over $\mb{Q}$ defined as
\begin{gather}
M_\alpha = \left(\begin{array}{cccccc}
0&0&0&\ldots&0&-c_0\\1&0&0&\ldots&0&-c_1\\0&1&0&\ldots&0&-c_2\\
\vdots&\vdots&\vdots&\ddots&\vdots&\vdots\\0&0&\ldots&1&0&-c_{d-2}\\0&0&0&\ldots&1&-c_{d-1}
\end{array} \right).
\end{gather}
\label{def:companion_matrix}
\end{definition}
\noindent One can show (see \cite{axler}) that 
\begin{gather}
\chi_{M_\alpha} = m_{M_\alpha} = m_\alpha.
\label{def:companion_matrix_equalities}
\end{gather}  
Assume $\alpha,\beta\in\mb{A}$ and their minimal polynomials are $m_\alpha$ and $m_\beta$ respectively. Companion matrix formalism allows us to find polynomials that annihilate  $\alpha\beta$ and $\alpha+\beta$ knowing polynomials $m_\alpha$ and $m_\beta$. To this end we construct matrices 
\begin{gather}
M_{\alpha\beta} = M_\alpha\otimes M_\beta,\quad M_{\alpha+\beta} = M_\alpha\otimes I_\beta+I_\alpha\otimes M_\beta.
\label{def:comp_matrix_identities}
\end{gather} 
The characteristic polynomials of these matrices annihilate $\alpha\beta$ and $\alpha+\beta$ respectively. We note, however,  that $m_{\alpha+\beta}$, $m_{\alpha\beta}$ may not be equal to characteristic polynomials of $M_{\alpha+\beta}$ and $M_{\alpha\beta}$ as these matrices are not companion matrices in the true sense of this word. Nevertheless, the equality holds when either $\alpha$ or $\beta$ belongs to $\mathbb{Q}$. 
\begin{fact}\label{fact-degree}
Assume that $\alpha\in\mb{Q}$ and $\beta\notin\mathbb{Q}$. Then $[\mathbb{Q}(\alpha+\beta):\mathbb{Q}]=[\mathbb{Q}(\alpha\beta):\mathbb{Q}]=[\mathbb{Q}(\beta):\mathbb{Q}]$ and the minimal polynomials of $\alpha\beta$ and $\alpha+\beta$ are given by the characteristic polynomials of 
\begin{gather}\label{simplified_comp}
M_{\alpha\beta} = \alpha M_\beta,\quad M_{\alpha+\beta} = \alpha I_\beta+ M_\beta,
\end{gather} 
where $M_\alpha$ and $M_\beta$ are the companion matrices of $\alpha$ and $\beta$.
\end{fact}
\begin{proof}
Note that $m_\alpha$ is the first order polynomial $m_\alpha(x)=x-\alpha$. Thus $M_\alpha$ is a $1\times 1$ matrix $M_\alpha=\alpha$. Using the companion matrix formalism we know that the characteristic polynomials of the matrices (\ref{simplified_comp}) annihilate $\alpha\beta$ and $\alpha+\beta$ respectively. We also know that $\mb{Q}(\beta)=\mb{Q}(\alpha+\beta)=\mb{Q}(\alpha\beta)$. Using formula (\ref{beta-deg}) we get that $\mr{deg}\;m_{\alpha+\beta}=\mr{deg}\;m_{\alpha\beta}=[\mb{Q}(\beta):\mb{Q}]=\mr{deg}({m_\beta})$. But the degrees of $\chi_{M_{\alpha\beta}}$ and $\chi_{M_{\alpha+\beta}}$ are also $\mr{deg}\;m_\beta$. The result follows.
\end{proof}
\noindent In general the degrees for $\alpha,\beta \in \mathbb{A}$ are bounded by 
\[
\max(\deg m_\alpha,\deg m_\beta)\leq \deg m_{\alpha+\beta,\alpha\beta} \leq \deg m_\alpha\deg m_\beta.
\] The lower bound corresponds to the case when one of the field extensions $\mb{Q}(\alpha),\;\mb{Q}(\beta)$ is a subfield of the second one, e.g. $\mb{Q}(\alpha)\subset\mb{Q}(\beta)$. Then both $\alpha+\beta$ and $\alpha\beta$ belong to the larger field extension. In that case the degrees of the minimal polynomials of $m_{\alpha+\beta}$ and $m_{\alpha\beta}$ are equal to the degree of $\mb{Q}(\beta)$ by definition. The other case is when $\mathbb{Q}(\alpha)$ and $\mb{Q}(\beta)$ have no common elements different than rational numbers. Then $\mb{Q}(\alpha\beta)$ and $Q(\alpha+\beta)$ are the field extensions of the degree at most $[\mb{Q}(\alpha):\mb{Q}]\cdot[\mb{Q}(\beta):\mb{Q}]$. In order to check it, let $\{1,\alpha,\ldots,\alpha^m \}$ be a basis of $\mb{Q}(\alpha)$ and  $\{1,\beta,\ldots,\beta^n \}$ be a basis of $\mb{Q}(\beta)$ such that $\alpha
^i\neq \beta^j$ for all $1\leq i\leq m$, $1\leq j\leq n$. Then the basis of $\mb{Q}(\alpha\beta)$ consists of all $\alpha^i\beta^j$'s and the maximal number of such \emph{different} basis elements is $[\mb{Q}(\alpha):\mb{Q}]\cdot[\mb{Q}(\beta):\mb{Q}]$. The same holds for $\mb{Q}(\alpha+\beta)$. Thus typically we need to find $\chi_{M_{\alpha\beta}}$ or $\chi_{M_{\alpha+\beta}}$  and factorize it over $\mathbb{Q}$ obtaining $\chi_{M_{\alpha+\beta,\alpha\beta}}=m_{\alpha+\beta,\alpha\beta}\cdot p_1\cdot\ldots\cdot p_k$. 
 
Polynomials with integer coefficients are characterised by special properties. Particularly important are the \emph{primitive} polynomials for which the greatest common divisor of their coefficients is equal to one.  The set of primitive polynomials is closed under multiplication in $\mb{Z}[x]$, in other words:
\begin{fact}\label{fact:monic}\cite{edwards}
Let $f,g\in\mb{Z}[x]$ be primitive polynomials. Then the product $f\cdot g$ is also a primitive polynomial.
\end{fact}
Every polynomial with rational coefficients $f\in \mb{Q}[x]$ can be associated with a primitive polynomial $\tilde{f}$ such, that
\begin{gather}\label{eq:primitive_associate}f(x)=\alpha \tilde{f},\:\alpha\in\mb{Q},
\end{gather}and the pair $(\tilde{f},\alpha)$ is unique. This fact allows us to formulate the following theorem that we prove for the reader's convenience.
\begin{lem}[Gauss lemma]\label{thm:gauss} \cite{edwards}
Let $f$ be a polynomial with integer coefficients. $f$ is reducible over $\mb{Q}$ if and only if it is reducible over $\mb{Z}$. 
\end{lem}
\begin{proof}The first part of the proof is trivial since $\mb{Z}\subset\mb{Q}$. Thus it is enough to prove that reducibility over $\mb{Q}$ implies reducibility over $\mb{Z}$.

Let $f\in \mb{Z}[x]$ be a primitive polynomial, reducible over $\mb{Q}$ (if not, there is always a pair $(\tilde{f},\alpha)$ such, that $\tilde{f}$ is primitive and $\alpha\in\mb{Z}$).
Assume that $f$ decomposes as $f=g\cdot h$, where $g,h\in\mb{Q}[x]$. By the identity (\ref{eq:primitive_associate}) these polynomials can be expressed as $g=\gamma\tilde{g}$, $h=\chi\tilde{h}$, where $\gamma,\chi\in\mb{Q}$ and $\tilde{g},\tilde{h}$ are primitive polynomials. Therefore we have $f=g\cdot h=\gamma\chi \tilde{g}\tilde{h}$, where $\tilde{g}\cdot\tilde{h}$ is a primitive polynomial.

Let us substitute $\gamma\chi=\frac{p}{q}$ for some $p,q\in\mb{Z}$, $\gcd(p,q)=1$. Hence 
\begin{gather}\label{eq:primitive_polynoms_equality}
qf=p\tilde{g}\cdot\tilde{h}.
\end{gather}Since $f$ and $\tilde{g}\cdot\tilde{h}$ are primitive, the greatest common divisor of the coefficients of the left and right hand side of (\ref{eq:primitive_polynoms_equality}) are equal $q$ and $p$ respectively. 
Thus $p=q$, which implies $\gamma\chi=1$ and $f=\tilde{g}\cdot\tilde{h}$ and $f=\tilde{g}\tilde{h}\in\mb{Z}[x]$. This completes the proof.
\end{proof}

\section{Trigonometric minimal polynomials }\label{app:cosines} 

By trigonometric polynomials we will understand the minimal polynomials for $\cos\phi,\sin\phi, 2\cos\frac{\phi}{2},\tan\phi$, where $\phi$ is a rational multiple of $\pi$. A detailed description of their properties was given in e.g. \cite{bayad,zeitlin,minpolys,minpolys2}. In this paper we will concentrate on minimal polynomials for $\cos\phi$ and $2\cos\frac{\phi}{2}$, that are closely related to Chebyshev polynomials of the first kind. 
\begin{definition}\label{def:chebyshev_def}
Chebyshev polynomials of the first kind $T_k(x)$, $k=0,1,\ldots$ are defined by the recurrence formula \begin{gather}\label{def:chebyshev}
T_k (x)= 2xT_{k-1}(x)-T_{k-2}(x),\:T_0(x)=1,\;T_1(x)=x.
\end{gather} Equivalently, $T_k(x)$'s are the polynomials  satisfying
\begin{gather}\label{def:chebyshev2} T_k(\cos\phi) = \cos k\phi.
\end{gather}
\end{definition} 
From the formula (\ref{def:chebyshev}) one can deduce properties of the coefficients of Chebyshev polynomials. 
\begin{fact}\label{fact:chebysh_coeffs}
Let $T_k(x)=\sum_{i=0}^{k} c_i x^i$ be the Chebyshev polynomial of the first kind of the degree $k$. The coefficients $c_0,\ldots,c_k$ satisfy:
\begin{enumerate}
\item $c_0,\ldots,c_k$ are integer numbers.
\item The leading coefficient $c_k$ is equal to $2^{k-1}$.
\item $T_k(x)$ is a polynomial of only odd/even powers of $x$ if $k$ is an odd/even number respectively.  
\item If $k$ is even, then the free term $c_0$ is given by $c_0=\pm 1$. 
\item If $k$ is odd, then the coefficients $c_0,c_1$ are equal $c_0=0$, $c_1=\pm k$.
\end{enumerate}
\end{fact}
Assume that $\phi$ is a rational multiple of $\pi$. In this case one can express the  minimal polynomial of $\cos\phi$ using the Chebyshev polynomials and vice versa. The explicit expressions are given in the following:
\begin{fact}\cite{zeitlin,minpolys}\label{def:psi}
A minimal polynomial $\psi_n(x)$ of $\cos\frac{2\pi}{n}$ is a polynomial defined as
\begin{gather}\label{eq:minpoly_12}
\psi_1(x)=x-\cos 2\pi=x-1,\quad \psi_2=x-\cos\pi = x+1,\\\label{exact_def_cos_minpoly}
\psi_n(x) = \prod_{\small \begin{array}{c}
1\leq k\leq \lfloor n/2\rfloor\\ \gcd(k,n)=1\end{array}  } \left(x-\cos\frac{2k\pi}{n} \right),\;n\geq 3,\\\deg\psi_n(x) = \left\lbrace
\begin{array}{cc}
1& n=1,2\\\frac{\varphi(n)}{2}&n\geq 3\end{array}\right.,
\end{gather}where $\varphi(n)$ is the \emph{Euler totient function} counting the positive integers $k\leq n$ that are coprime to $n$. The canonical identity relating $\psi_n(x)$ with Chebyshev polynomials of the first kind is of the form 
\begin{gather}\label{eq:zeitlin_definition_odd}
2^k\prod_{d|n}\psi_d(x) =  T_{k+1}(x)-T_k(x),\:\:n=2k+1,\\
2^k\prod_{d|n}\psi_d(x) =  T_{k+1}(x)-T_{k-1}(x), \:\:n=2k.\label{eq:zeitlin_definition_even}
\end{gather}
%
\end{fact}
\begin{rem}
Note that the sum in (\ref{exact_def_cos_minpoly}) is over $1\leq k\leq \lfloor n/2\rfloor$, $\gcd(k,n)=1$ instead of $1\leq k\leq n$. It stems from the symmetry of cosine, {\it i.e.} $\cos\phi=\cos(-\phi)$. Note that for every $\phi=\frac{2k\pi}{n}$ the opposite angle is defined as $-\phi=-\frac{2k\pi}{n}=\frac{2(n-k)\pi}{n}$. By properties of the greatest common divisor \cite{hardy}, $\gcd(k,n)=1$ implies $\gcd(n-k,n)=1$. Thus every root $\cos\frac{2k\pi}{n}$, where $1\leq k\leq \lfloor n/2\rfloor$, is equal to the root $\cos\frac{2(n-k)\pi}{n}$ where $\lfloor n/2\rfloor< n-k<n$. Therefore $\cos\frac{2k\pi}{n}$ for $1\leq k\leq \lfloor n/2\rfloor$ are all the possible roots of $\psi_n(x)$.
\end{rem}\noindent
Before we formulate the lemma describing a important property of $\psi_n(x)$  we will present definition of the M\"obius function that allows us to invert (\ref{eq:zeitlin_definition_even}) and (\ref{eq:zeitlin_definition_odd}) and express $\psi_n(x)$ as a function of Chebyshev polynomials \cite{bayad}.  
\begin{definition}\label{def:mobius}
M\"obius function is an integer function taking three possible values $\mu(n) = \{-1,0,1 \}$. It depends on the prime factorisation of $k$ (this means a unique representation of $n$ as a product of prime numbers). Let $n =p_1^{n_1}\ldots p_m^{n_m}$, then (1) $\mu(n)=0$ if at least one $n_i$ is larger than one, (2) $\mu(n)=-1$ if $m$ is an odd number and $n_1=\ldots=n_m=1$ (we say that $n$ is \emph{square free}), (3) $\mu(n)=1$ if $m$ is even and $n$ is square free. \newline\noindent The M\"obius function satisfies the canonical identity \cite{hardy}
\begin{gather}\label{mobius_function_sum}
\sum_{d|n}\mu(d) = \left\lbrace\begin{array}{cc}
0,&n>1\\1,&n=1
\end{array}\right..
\end{gather}
\end{definition}
\begin{fact}\cite{bayad}\label{lem:minpoly_inverse} For $n\geq 3$, odd and $m\geq1$, $m,n\in\mb{N}$ the M\"obius inverses of (\ref{eq:zeitlin_definition_odd}) and (\ref{eq:zeitlin_definition_even}), given by the M\"obius inversion formula, are the following
\begin{gather}\label{eq:zeitlin_inverse_odd}\psi_n(x)=  \prod_{d|n}\left[ 2^{-\lfloor d/2\rfloor}\left(T_{\lfloor d/2\rfloor+1}(x)-T_{\lfloor d/2\rfloor}(x \right)\right] ^{\mu(n/d)},\\\label{eq:zeitlin_inverse_even}
\psi_{2^mn}(x)=\prod_{d|n}\left[\frac{2^{-\lfloor 2^{m-1}d\rfloor}\left(T_{ \lfloor 2^{m-1}d\rfloor+1 }(x)-T_{\lfloor 2^{m-1}d\rfloor-1}(x \right)  }{2^{-\lfloor 2^{m-2}d\rfloor }\left(T_{ \lfloor 2^{m-2}d\rfloor+1 }(x)-T_{\lfloor 2^{m-2}d\rfloor-1}(x \right) } \right]^{\mu(n/d)},\quad m>1,\\
\label{eq:zeitlin_inverse_even1}\psi_{2^mn}(x)=\prod_{d|n}\left[\frac{2^{- d}\left(T_{ \lfloor d\rfloor+1 }(x)-T_{\lfloor d\rfloor-1} (x)\right)  }{2^{-\lfloor d/2\rfloor }\left(T_{ \lfloor d/2\rfloor+1 }(x)-T_{\lfloor d/2\rfloor} (x)\right) } \right]^{\mu(n/d)},\quad m=1.
\end{gather}
\end{fact}\noindent
The following property of $\psi_n(x)$ is crucial in proving theorem \ref{thm:main}.
\begin{lem} \label{fact:coeffs_minpoly_cos}
At least one coefficient of $\psi_n(x)$ is non-integer if $n\notin\{1,2,4\}$. 
\end{lem}\begin{proof}Note that $\cos\frac{2k\pi}{n}$, $\gcd(k,n)=1$ are integers for $n\in\{1,2,4 \}$, thus the corresponding minimal polynomials $\psi_n(x)$ belong to $\mb{Z}[x]$. In other cases one can distinguish the following situations: 1) $n$ is an odd prime number, 2) $n$ is an odd composite number, 3) $n$ is an even composite number.
\begin{enumerate}
\item In the case 1) $n$ has exactly two divisors and the formula (\ref{eq:zeitlin_definition_odd}) simplifies to
\begin{gather}\label{eq:two_product}
\psi_1(x)\psi_{n}(x)=(x-1)\psi_{n}(x)=2^{-\lfloor n/2\rfloor}  \left(T_{\lfloor n/2\rfloor+1}(x)-T_{\lfloor n/2\rfloor}(x) \right).
\end{gather}Note that $T_{\lfloor n/2\rfloor+1}(x)-T_{\lfloor n/2\rfloor}(x)$ is a difference of polynomials of the even and the odd degrees, therefore by fact \ref{fact:chebysh_coeffs} it has a free term equal $\pm 1$ hence the free term of the right hand side is $\pm\frac{1}{2^{\lfloor n/2\rfloor}}$. Note that the free term of left hand side is determined by the free term of $\psi_n(x)$. Comparing the left and right side of (\ref{eq:two_product}) one can see that the free term of $\psi_n(x)$, {\it i.e.} $c_0=\pm\frac{1}{2^{\lfloor n/2\rfloor}}$, must be a rational number since we consider $n\leq 3$.
\item In the case 2) we will prove that $\psi_n(x)$ has a non-integer free term by applying formula (\ref{eq:zeitlin_inverse_odd}) and using properties of the M\"obius function. Let $n=p_1^{n_1}\ldots p_m^{n_m}$ be the prime factorisation of $n$. Let $\mathcal{D}_n^+$ and $\mathcal{D}_n^-$ be the sets of all square free (as in definition \ref{def:mobius}) divisors of $n$ that have even and odd number of prime divisors respectively. By (\ref{mobius_function_sum}) we have $|\mathcal{D}_{n}^+|=|\mathcal{D}_{n}^-|$. Note that the free term of the right hand side of (\ref{eq:zeitlin_inverse_odd}) can be written in the form
\begin{gather}\label{eq:coeff_co_case2}
c_0=\pm \prod_{d|n} \left(\frac{1}{2^{\lfloor d/2\rfloor}} \right)^{\mu(n/d)}  =\pm \prod_{\frac{n}{d}\in \mathcal{D}_n^+\cup\mathcal{D}_n^-} \left(\frac{1}{2^{\lfloor d/2\rfloor}} \right)^{\mu(n/d)}=\frac{2^{\sum_{\frac{n}{d}\in  \mathcal{D}_n^-} \lfloor d /2  \rfloor }}{2^{\sum_{\frac{n}{d}\in  \mathcal{D}_n^+} \lfloor d/2\rfloor} } .
\end{gather}
We next raise $c_0$ to the power $p=2p_1\cdot\ldots \cdot p_m$


\begin{gather}\label{eq:co_neq}
c_0^{p}= \frac{2^{\sum_{\frac{n}{d}\in  \mathcal{D}_n^-} p\lfloor d /2  \rfloor }}{2^{\sum_{\frac{n}{d}\in  \mathcal{D}_n^+} p\lfloor d/2\rfloor} } =\frac{2^{\sum_{\frac{n}{d}\in  \mathcal{D}_n^-} p( d-1 )/2   }}{2^{\sum_{\frac{n}{d}\in  \mathcal{D}_n^+} p( d-1) /2  } } =\frac{ 2^{ -\frac{pk}{2}+\sum_{\frac{n}{d}\in  \mathcal{D}_n^-} \frac{pd}{2} }}{2^{ -\frac{pk}{2}+\sum_{\frac{n}{d}\in  \mathcal{D}_n^+} \frac{pd}{2} }}= \frac{2^{\sum_{\frac{n}{d}\in  \mathcal{D}_n^-} \frac{pd}{2} }}{2^{\sum_{\frac{n}{d}\in  \mathcal{D}_n^+} \frac{pd}{2}}}.
\end{gather}  

Note that $c_0$ is non-integer if and only if $c_0^p<1$. In order to find the appropriate condition we use the fact  that if $x=\frac{n}{d}$ belongs to $\mc{D}_n^-$ or $\mc{D}_n^+$ then 
\begin{gather}\label{eq:delta_i} 
p\cdot d=p \frac{n}{x}=2n\cdot y,
\end{gather}
where $y$ is the square-free product of such prime divisors of $n$ that do not appear in prime factorisation of $x$. In particular, if $m$ is an even number and $x\in\mc{D}_n^+$, then $y$ must also belong to $\mc{D}_n^+$.  Similarly $x\in\mc{D}_n^-$ implies that $y\in\mc{D}_n^-$. When $m$ is odd one easily checks that  $x\in\mc{D}_n^+$ implies that $y\in\mc{D}_n^-$ and {\it vice versa}.

Let us consider the case when $m$ is an even number. Taking $pd=2ny$ one can rewrite (\ref{eq:co_neq}) as
\begin{gather}\label{eq:transformed_even}
c_0^p=\frac{2^{\sum_{\frac{n}{d}\in  \mathcal{D}_n^-} \frac{pd}{2} }}{2^{\sum_{\frac{n}{d}\in  \mathcal{D}_n^+} \frac{pd}{2}}}=
\frac{2^{\sum_{y\in  \mathcal{D}_n^-}n y }}{2^{\sum_{y\in  \mathcal{D}_n^+} ny}}=
\frac{2^{n\sum_{y\in  \mathcal{D}_n^-} y }}{2^{n\sum_{y\in  \mathcal{D}_n^+} y}}.
\end{gather}
As one can easily see $c_0^p<1$ if the following holds
\begin{gather}
\sum_{y\in  \mathcal{D}_n^+} y-\sum_{y\in  \mathcal{D}_n^-} y>0.
\end{gather}
Note that this expression is equivalent to the product
\begin{gather}
\sum_{y\in  \mathcal{D}_n^+} y-\sum_{y\in  \mathcal{D}_n^-} y=(1-p_1)\ldots(1-p_m),
\end{gather}which is always larger than zero if $m$ is  even, thus $c_0$ is non-integer in this case.

Next let us consider $n$ such, that $m$ is an odd number. Doing {\it mutatis mutandis} to the case when $m$ is even one can transform (\ref{eq:co_neq}) to the form

\begin{gather}\label{eq:transformed_odd}
c_0^p=\frac{2^{\sum_{\frac{n}{d}\in  \mathcal{D}_n^-} pd }}{2^{\sum_{\frac{n}{d}\in  \mathcal{D}_n^+} pd}}=
\frac{2^{\sum_{y\in  \mathcal{D}_n^+}n y }}{2^{\sum_{y\in  \mathcal{D}_n^-} ny}}=
\frac{2^{n\sum_{y\in  \mathcal{D}_n^+} y }}{2^{n\sum_{y\in  \mathcal{D}_n^-} y}},
\end{gather}thus the condition for $c_0^p<1$ is given by 
\begin{gather}
\sum_{y\in  \mathcal{D}_n^+} y-\sum_{y\in  \mathcal{D}_n^-} y=(1-p_1)\ldots(1-p_m)<0.
\end{gather}Note that this  is always satisfied if $m$ is odd, which means that $c_0^p$ is indeed smaller than one. This way we have proven lemma \ref{fact:coeffs_minpoly_cos} for the case when $n>1$ and $n$ is odd.

\item In  case 3) we will use the similar approach as for  case 2). Recall that every even number $n$ can be represented as $n=2^\eta k$, where $k$ - odd and $\eta\in\mb{Z}_+$. This allows us to define $\psi_n(x)$ by the formula (\ref{eq:zeitlin_inverse_even}). Let us define the sets $\mc{D}_n^+$ and $\mc{D}_n^-$ like previously. Using fact \ref{fact:chebysh_coeffs} one can extract the free terms from (\ref{eq:zeitlin_inverse_even}) and (\ref{eq:zeitlin_inverse_even1}) as 
\begin{gather}\label{eq:coeff_even_1}
c_0=\prod_{d|k}\left(\frac{2}{2^{2^{\eta-2}d}}\right)^{\mu(k/d)}=\frac{2^{2^{\eta-2}\sum_{\frac{n}{d}\in\mc{D}_n^- } d}}{2^{2^{\eta-2}\sum_{\frac{n}{d}\in\mc{D}_n^+ } d}},\quad \eta>1,\\
\label{eq:coeff_even_2} c_0=\prod_{d|k}\left(\frac{1}{2^{\lfloor d/2\rfloor}}\right)^{\mu(k/d)}=\frac{2^{ \sum_{\frac{n}{d}\in\mc{D}_n^- } \lfloor d/2\rfloor}}{2^{\sum_{\frac{n}{d}\in\mc{D}_n^+ } \lfloor d/2\rfloor }},\quad \eta=1.
\end{gather}By properties of the M\"obius function all the sums and products are over the same number of divisors. Note that (\ref{eq:coeff_even_2}) and (\ref{eq:coeff_co_case2}) are exactly equal, whereas (\ref{eq:coeff_even_1}) has a very similar form to (\ref{eq:coeff_co_case2}). Using very similar  reasoning as in the case 2) we obtain immediately that $c_0$ is non-integer unless $n=2,4$,
which completes the proof.  
\end{enumerate}
\end{proof}
\paragraph*{Example:} In this paragraph we will illustrate the method presented in the proof of lemma \ref{fact:coeffs_minpoly_cos} with the example for $n=15$. First, $\mc{D}^+_n=\{ 1,15\}$ and $\mc{D}^-_n=\{ 3,5\}$. The coefficient $c_0$ of $\psi_{15}(x)$ is given by
\begin{gather}\label{eq:example1}
c_0=\frac{2^{\lfloor 5/2\rfloor+\lfloor 3/2\rfloor}}{2^{\lfloor 15/2\rfloor+\lfloor 1/2\rfloor}}.
\end{gather}
Using the reasoning presented for the case 2) we raise $c_0$ to the power $p=30$ and obtain the condition
\begin{gather*}
c_0^{30}=\frac{2^{15(5+3)-15}}{2^{15(15+1)-15}}=\frac{2^{15\sum_{y\in\mc{D}_{15}^- } y}}{2^{15\sum_{y\in\mc{D}_{15}^+ } y}}<1\Leftrightarrow \sum_{y\in\mc{D}_{15}^+ } y-\sum_{y\in\mc{D}_{15}^- } y=16-8>0.
\end{gather*}Thus we have shown that $c_0<1$.  On the other hand one can see immediately from definition (\ref{eq:example1}) than $c_0$ for $\psi_{15}(x)$ is equal to $2^{-8}\notin\mb{Z}$.

The proof of theorem \ref{thm:main} presented in this paper is based also on properties of minimal polynomials for $2\cos\frac{\pi}{n}$. However, before we describe them, we will recall the notion of cyclotomic polynomials. Since $2\cos\frac{k\pi}{n}$ is a sum of two roots of unity $e^{ik\pi/n}$ and $e^{-ik\pi/n}$ one can conclude that properties of minimal polynomials for double cosines depend on properties of cyclotomic polynomials.
\begin{definition}\label{def:cyclotomic}
A cyclotomic polynomial $\Phi_n(x)$ is the polynomial with integer coefficients, irreducible over $\mb{Q}$, defined as 
\begin{gather}\label{eq:cyclotomic_def}\Phi_n(x)=\prod_{\small \begin{array}{c}
 1\leq k\leq n\\\gcd(k,n)=1
\end{array} }\left(x-e^{2i\pi k/n} \right),
\end{gather}and satisfying the identity $ 1-x^n=\prod_{d|n}\Phi_d(x)$. $\Phi_n(x)$ is the minimal polynomial of an $n$-th root of identity. 
\end{definition}\noindent
The basic facts concerning minimal polynomials for double cosines are the following:
\begin{lem}\cite{bayad,minpolys2}\label{lem:minpoly2}
The minimal polynomial for $2\cos\frac{\pi}{n}$ is a polynomial defined as
\begin{gather}
\eta_n(x) = 2^{\deg\psi_{2n}(x)}\psi_{2n}\left(\frac{x}{2}\right).
\label{def:minpoly_2cosphi}
\end{gather}
Using fact \ref{def:psi} one can write it explicitly as:  \begin{gather}\label{eq:def_minpolys2cos}
\eta_1(x)=x+2,\\
\eta_n(x)=\prod_{\small\begin{array}{c}
1\leq k\leq n\\ \gcd(k,2n)=1
\end{array}}\left(x-2\cos\frac{k\pi}{n} \right),\:n\leq 2
\end{gather}

\end{lem}
\begin{proof}
Note that for arbitrary $1\leq k\leq n $, s.t. $\gcd(k,2n)=1$, $\cos\frac{k\pi}{n}$ is a root of the polynomial $\psi_{2n}(x)$ of a degree $d=\deg\psi_{2n}(x)=\frac{\varphi(2n)}{2}$ with the coefficients $c_0,\ldots,c_{d-1},c_d$. Since $2\cos\frac{\pi}{n}$ is a product of $2$ and a root of $\cos\frac{k\pi}{n}$, the companion matrix $M_{2\cos\frac{k\pi}{n}}$ is given by $M_{2\cos\frac{k\pi}{n}} = 2\cdot M_{\cos\frac{k\pi}{n}}$ (\ref{def:comp_matrix_identities}) and by fact \ref{fact-degree} the minimal polynomial $\eta_n(x)$ is equal to the characteristic polynomial 
$\chi_{M_{2\cos\frac{k\pi}{n}}}$. Let us write this matrix explicitly. The polynomial $\eta_n(x)$ is given by
\begin{gather}\label{def:minpoly2cosphi_matrix}
\eta_n(x)=\chi_{M_{2\cos\frac{k\pi}{n}}} =\det \left(\begin{array}{cccc}-x&0&\ldots&-2c_0\\2&-x&\ldots&-2c_1\\0&2&\ldots&-2c_2\\\vdots&\vdots&\ddots&\vdots\\0&\ldots&2&-2c_{d-1}-x\end{array} \right).
\end{gather} 
The substitution $x\to 2y$ transforms (\ref{def:minpoly2cosphi_matrix}) to the form $\eta_n(2y)=\chi_{2 M_{\cos\frac{k\pi}{n}}}$. Since $M_{\cos\frac{k\pi}{n}}$ is a $d\times d$ matrix we arrive at
\begin{gather} 
\eta_n(2y)=\det 2(M_{\cos\frac{k\pi}{n}}-yI) =2^d \psi_{2n}(y) =2^d \psi_{2n}\left(\frac{x}{2} \right),
\end{gather}
which is exactly (\ref{def:minpoly_2cosphi}). The proof is complete
\end{proof}
\noindent Another property of $\eta_n(x)$ that is crucial for the proof of theorem \ref{thm:main} is proven in the following. 
\begin{lem}\label{fact:minpoly_2cos_coeffs}
All the coefficients of $\eta_n(x)$ are integers.
\end{lem}
\begin{proof}
In order to prove this fact we use lemma \ref{thm:gauss} and properties of cyclotomic polynomials. First, note that $2\cos\frac{k\pi}{n}=2\cos\frac{2k\pi}{2n},\;\gcd(k,2n)=1$ can be written as a sum $e^{ik\pi/n}+e^{- ik\pi/n}$, where both $e^{ik\pi/n}$ and $e^{-ik\pi/n}$ are roots of the cyclotomic polynomial $\Psi_{2n}(x)$, whose companion matrix will be denoted by $M_{\Psi_{2n}(x)}$. Since $\Psi_{2n}(x)$ is a polynomial with integer coefficients, $M_{\Psi_{2n}(x)}$ has only integer entries. 

Recall that using (\ref{def:comp_matrix_identities})  $2\cos\frac{\pi}{n}$ is a root of the characteristic polynomial of the matrix defined as 
\begin{gather}
M=M_{\Psi_{2n}(x)}\otimes I+I\otimes M_{\Psi_{2n}(x)},
\end{gather}and $\chi{_M}$ belongs obviously to $\mb{Z}[x]$. The degree of the minimal polynomial of $2\cos\frac{\pi}{n}$ is smaller than the degree of $\chi_{M_{2\cos\frac{\pi}{n}}}$ as $\mb{Q}(2\cos\frac{k\pi}{n})\subset\mb{Q}(e^{ik\pi/n})$. This means, the characteristic polynomial of $M$ can be factorized as a product of at least two polynomials $$\chi_{M}=\eta_{n}(x)\cdot p(x),\:\:\mr{where\:in\:general}\:\: p(x),\eta_n(x)\in \mathbb{Q}[x].$$Since $\chi_{M}$ belongs to  $\mb{Z}[x]$, by lemma
\ref{thm:gauss} the polynomials $\eta_{n}(x)$ and $p(x)$ have integer coefficients. As $\chi_{M}$ is monic, $p(x)$ and $\eta_{n}(x)$ are also monic and the result follows. 
\end{proof}
\section{The proof of  theorem  1}\label{sec:proofs}
The aim of this section is to decide, using the tools presented in previous sections, for which $\phi=\frac{2k\pi}{n}$ the left hand side of the equation 
\begin{gather}\label{eq:gamma_cos_plus_1}2\cos\frac{\gamma}{2}=\cos\phi+1,
\end{gather}
is a root of the minimal polynomial $\eta_m(x)$ for some $m\in\mb{N}$. 


The minimal polynomial $m_{\cos\phi+1}(x)$ can be found using the companion matrix formalism explained in section \ref{sec:comp_matrix}. It follows that $\cos\phi+1$ is a root of the characteristic polynomial of the matrix $M_{\cos\phi+1}=M_{\cos\phi}+ I$, where
\begin{gather*}
M_{\cos\phi} = \left(\begin{array}{cccc}
0&0&\ldots&-c_0\\1&0&\ldots&-c_1\\0&1&\ldots&-c_2\\\vdots&\vdots&\ddots&\vdots\\0&\ldots&0&-c_{d-1}
\end{array} \right),
\end{gather*}
Note that the field extensions $\mb{Q}({2\cos\frac{\gamma}{2}})=\mathbb{Q}(\cos\phi+1)$ and $\mathbb{Q}(\cos\phi)$ are of the same algebraic degree (see fact \ref{fact-degree}), therefore the characteristic polynomial of $M_{\cos\phi+1}$ is exactly the minimal polynomial $m_{2\cos\frac{\gamma}{2}}(x)$. One can compute $\chi_{M_{\cos\phi+1}}$ as the determinant of the following matrix:
\begin{gather}
 M_{\cos\phi+1}-Ix = \left(\begin{array}{cccccc}
1-x&0&0&\ldots&0&-c_0\\1&1-x&0&\ldots&0&-c_1\\0&1&1-x&\ldots&0&-c_2\\
\vdots&\vdots&\vdots&\ddots&\vdots&\vdots\\0&0&\ldots&1&1-x&-c_{d-2}\\0&0&\ldots&0&1&1-c_{d-1}-x
\end{array} \right).
\label{def:comp_matrox_general}
\end{gather} Expansion with respect to the first row gives us to the following expression:
\begin{gather}
m_{2\cos\frac{\gamma}{2}}(x)=\chi_{M_{\cos\phi+1}}(x)=\det(M_{\cos\phi+1}-Ix)=\sum_{i=0}^d \omega_i \cdot x^i=\\ = -\left[\sum_{i=0}^{d-3} c_i (-1)^{d+i+1} (1-x)^i \right] + (1-x)^{d-2}[x^2+x(c_{d-1}-2)+c_{d-2}-c_{d-1}+1].
\label{def:companion_determinant}
\end{gather}
One can simplify (\ref{def:companion_determinant}) using the binomial formula. As a result we obtain the following relations between the coefficients of $\psi_n(x)$ and $\chi_{M_{\cos\phi+1}}$:
\begin{equation}
\begin{array}{c}
\omega_0 = \sum_{i=0}^{d-3}c_i(-1)^{d+i} +( c_{d-2}-c_{d-1}+1)(-1)^{d}, \\
\omega_1 = \sum_{i=1}^{d-3} c_i(-1)^{d+1+i}  +{d-2\choose 0} (c_{d-1}-2) - {d-2\choose 1}(c_{d-2}-c_{d-1}+1),\\
\omega_2 = \sum_{i=2}^{d-3}c_i(-1)^{d+i+2} {i\choose 2} +{d-2\choose 0}-{d-2\choose 1}(c_{d-2}+1)+{d-2\choose 2}(c_{d-2}-c_{d-1}+1),\\
\vdots\\
\omega_{k} = \sum_{i=k}^{d-3}c_i(-1)^{d+i+k} {i\choose k}+{d-2\choose k}(-1)^k(c_{d-2}-c_{d-1}+1)+{i\choose k-1}(-1)^{k-1}(c_{d-1}-2)+{d-2\choose k-2}(-1)^{k-2}  \\\vdots\\
\omega_{d-3}=-c_{d-3}+(c-2)(-1)^{d-3}(c_{d-2}-c_{d-1}+1)
+ {d-2 \choose d-4}(-1)^{d-4}(c_{d-1}-2)+(-1)^{d-5}{d-2\choose d-5}, \\
\omega_{d-2} = (-1)^{d-2}(c_{d-2}-c_{d-1}+1)+(d-2)(-1)^{d-3}(c_{d-1}-2)+(-1)^{d-4}{d-2\choose d-4},\\
\omega_{d-1} = (-1)^{d}(c_{d-1}-2)+(d-2)(-1)^{d-1},\\
\omega_d = (-1)^{d}.
\end{array}
\label{companion_determinant_coeffs}
\end{equation}
Recall that by lemma \ref{lem:minpoly2} the coefficients of $m_{2\cos\frac{\gamma}{2}}(x)$ must be all integers if $\gamma$ is a rational multiple of $\pi$. One has to check if this condition is satisfied by all $\omega_i$'s starting from $\omega_{d-1}$. For this coefficient we have:
\begin{displaymath}
(-1)^{d-2}(c_{d-1}-2)+(d-2)(-1)^{d-3}\in\mathbb{N}\Rightarrow c_{d-1}\in\mathbb{N}.
\end{displaymath}
Note that $\omega_{d-1}\notin\mathbb{N}$ if and only if $c_{d-1}\notin\mathbb{N}$. 
In this case we are done, however there are polynomials $\psi_n(x)$ for which $c_{d-1}\in\mb{N}$. In this case we consider the equation for $\omega_{d-2}$: 
\begin{displaymath}
 \omega_{d-2} = (-1)^{d-2}(c_{d-2}-c_{d-1}+1)+(d-2)(-1)^{d-3}(c_{d-1}-2)+(-1)^{d-4}{d-2\choose d-4}\in\mathbb{N}\Rightarrow c_{d-2}\in\mathbb{N}.
\end{displaymath}Assuming that $c_{d-1}\in\mb{Z}$, the coefficient $\omega_{d-2}$ is non-integer only if $c_{d-2}\notin\mathbb{Z}$. One can use the same reasoning for the other $\omega_i$'s step by step and notice from (\ref{companion_determinant_coeffs}) that each $\omega_i$ depends on the coefficients $c_{d-1},\ldots,c_{i+1},c_i$, where $c_i$ is multiplied by the factor $(-1)^{d+2i} {i\choose i}=\pm 1$. If all of the coefficients  $c_{d-1},\ldots,c_{i+1},c_i$ are integers, then also $\omega_i$ belongs to $\mathbb{Z}$. Hence all $\omega_1,\ldots,\omega_d$ are integers if and only if  $c_1,c_2,\ldots,c_{d-1},c_d\in\mb{Z}$. On the other hand we have shown in lemma \ref{fact:coeffs_minpoly_cos} that $\psi_n(x)$ has always at least one non-integer coefficient if $n\notin\{1,2,4\}$. Therefore at least one $\omega_i$ does not belong to $\mb{Z}$ and $m_{2\cos\frac{\gamma}{2}}$ cannot be the minimal polynomial $\eta_m(x)$ for any $m\in\mb{Z}$. This means, $\gamma$ must be an irrational multiple of $\pi$.

Finally we consider what happens in the exceptional cases when $\phi=\{ \frac{\pi}{2},\pi,\frac{3\pi}{2},2\pi\}$, {\it i.e.} when $n\in\{1,2,4\}$. 
\begin{enumerate}
\item Let $\phi=\pm\frac{\pi}{2}$. The corresponding minimal polynomial and its companion matrix are of the form $\psi_{4}(x)=x$ and $M_{\cos\frac{\pi}{2}}=0$. Thus $M_{\cos\phi+1} = 1$ and its minimal polynomial is $m_{M_{\cos\phi+1}}(x)=x-1$. We have that $x-1 = \eta_3(x),$ thus $\gamma = \frac{k\pi}{3}$. 
\item Let $\phi=2\pi$, then $\psi_{1}(x)=x-1$, $M_{\cos 2\pi}=1$. The companion matrix for $M_{\cos\phi+1}$ is a $1\times 1$ matrix equal to $ M_{\cos\phi+1}= 1+1 = 2$. Thus $m_{M_{\cos\phi+1}}(x)=x-2$. This polynomial corresponds to $\gamma = 2\pi$.
\item Assume $\phi=\pi$, then $\psi_{2}(x)=x+1$, $M_{\cos\pi}=-1$ and $M_{\cos\phi+1} = 1-1 = 0$. Thus the minimal polynomial of $m_{\cos\phi+1}(x)=x=\eta_{2}(x)$ which is the minimal polynomial of  $2\cos\gamma/2$ for $\gamma =\pi$.
\end{enumerate}

\section{More examples}
Assume now that we have two finite order $SU(2)$-gates. They can be written as $U(\phi_1,\vec{k}_1)$ and $U_2(\phi_2,\vec{k}_2)$ where
\begin{gather}
U(\phi,\vec{k})=\cos\phi I+\sin\phi\left(k_xX+k_yY+k_zZ\right),\,\,\vec{k}=[k_x,k_y,k_z],\,\,k_x^2+k_y^2+k_z^2=1,
\end{gather}
and both $\phi_1$ and $\phi_2$ are rational multiples of $\pi$. Their product is $U(\gamma,\vec{k}_\gamma)=U(\phi_1,\vec{k}_1)U_2(\phi_2,\vec{k}_2)$, where
\begin{gather}
\cos\gamma=\cos\phi_1\cos\phi_2+\sin\phi_1\sin\phi_2\vec{k}_1\cdot\vec{k}_2.
\end{gather}
Under our assumptions $\cos\phi_1$, $\cos\phi_2$ and $\sin\phi_1$, $\sin\phi_2$ are algebraic numbers. It is well known that the sum and the product of an algebraic number and a transcendental number is transcendental. Moreover, if $\cos\phi$ is transcendental then $\phi$ is an irrational multiple of $\pi$. Thus if $\cos\theta =\vec{k}_1\cdot\vec{k}_2$ is transcendental then $U(\gamma,\vec{k}_\gamma)$ is of infinite order. In the following we show that it can also happen when $\cos\theta$ is an algebraic number.

\subsection{Gates $H$ and $T(\phi)$}
Let us consider a popular in quantum computation set of unitary gates  $\mc{S}=\{H,T\left(\phi\right)\} $, where $H=U(\frac{\pi}{2},\frac{1}{\sqrt{2}}(\vec{k}_z+\vec{k}_y))$ is the Hadamard gate, and $T\left(\phi\right)=U(\phi,\vec{k}_z)$ is a phase gate. Explicit matrices are given by:
\begin{gather}
H=\frac{1}{\sqrt{2}}\left(\begin{array}{cc}1&1\\-1&1
\end{array} \right),\:T\left(\phi\right)=\left(\begin{array}{cc}e^{i\phi}&0\\0&e^{-i\phi}
\end{array} \right).
\end{gather}We look for angles $\phi$ for which the product $HT_\phi$ is of infinite order. An example value of such $\phi$, $\phi=\frac{\pi}{8}$, was given in \cite{Boykin}.  A more general answer can be found using the method presented in previous sections. The matrices $H$ and $T(\phi)$ represented as rotation matrices from $SO(3)$ are of the form
\begin{gather}\label{eq:h_t_rotation_definition} \mr{Ad}_H=\left(\begin{array}{ccc}-1&0&0\\0&0&1\\0&1&0
\end{array}\right),\:\:\mr{Ad}_{T(\phi)}=\left(\begin{array}{ccc}\cos 2\phi&\sin 2\phi&0\\-\sin 2\phi&\cos 2\phi&0\\0&0&1.
\end{array}\right)
\end{gather}The product of $\mr{Ad}_H$ and $\mr{Ad}_{T(\phi)}$ is a rotation by angle $\gamma$ given by
\begin{gather*}
\mr{tr}\mr{Ad}_H\mr{Ad}_{T(\phi)}=\mr{tr}O(\gamma,\vec{k}),\\-\cos 2\phi=2\cos\gamma+1\Rightarrow -2\cos\gamma=\cos 2\phi+1,
\end{gather*}but from trigonometric identities we have that $-\cos\gamma=\cos(\gamma+\pi)$ thus we can write
\begin{gather}\label{eq:h_t_trace_formula}
2\cos\gamma'=\cos 2\phi+1,\:\:\gamma'=\gamma+\pi.
\end{gather}Note that equation (\ref{eq:h_t_trace_formula}) is of the same form as 
(\ref{eq:gamma_cos_plus_1}), thus by theorem \ref{lem:productRotations} we can say immediately that $\mr{Ad}_H\cdot\mr{Ad}_{T(\phi)}$ is of infinite order if and only if
\begin{gather}\label{eq:phi_hadd_cond} 
2\phi\neq\{0,\pi,\pm\frac{\pi}{2} \}.
\end{gather}
Hence the gate $HT(\phi)$ is of infinite order if and only if
\begin{gather}
\phi\neq\{0,\pi,\pm\frac{\pi}{2},\frac{\pi}{4},\frac{3\pi}{4} \}.
\end{gather}

\section{Acknowledgment} 
This work was supported by National Science Centre, Poland under the grant SONATA BIS: 2015/18/E/ST1/00200.

\end{document}